\documentclass{llncs}

\usepackage{amsmath}
\usepackage{amssymb}
\usepackage{graphics}
\usepackage{stmaryrd}
\usepackage{url}

\urldef{\mailrb}\path|radim.belohlavek@acm.org|
\urldef{\maillu}\path|lucie.urbanova01@upol.cz|
\urldef{\mailvv}\path|vychodil@acm.org|

\newcommand{\project}[2]{\ensuremath{\pi_{#1}(#2)}}

\newcommand{\select}[2]{\ensuremath{\sigma_{#1}(#2)}}
\newcommand{\selectH}[2]{\ensuremath{\sigma^\approx_{#1}(#2)}}
\newcommand{\join}[3]{\ensuremath{#2\bowtie_{#1}#3}}

\def\hh{\ast}

\def\atr#1{\ensuremath{\text{\textit{\texttt{\MakeLowercase{#1}}}}}}
\def\val#1{\ensuremath{\text{\texttt{#1}}}}

\begin{document}


\mainmatter

\title{Sensitivity Analysis for Declarative Relational Query Languages
  with Ordinal Ranks\thanks{%
    Supported by grant no. P103/11/1456 of the Czech Science Foundation.}%
  \thanks{%
    The paper will appear in 
    Proceedings of the 19th International Conference on
    Applications of Declarative Programming and Knowledge Management
    (INAP 2011).}}

\author{Radim Belohlavek \and Lucie Urbanova \and Vilem Vychodil}

\institute{DAMOL (Data Analysis and Modeling Laboratory) \\
  Dept. Computer Science, Palacky University, Olomouc \\
  17. listopadu 12, CZ--77146 Olomouc, Czech Republic \\
  \mailrb, \maillu, \mailvv}

\maketitle


\begin{abstract}
  We present sensitivity analysis for results of query executions in a relational
  model of data extended by ordinal ranks. The underlying model of data results
  from the ordinary Codd's model of data in which we consider ordinal ranks
  of tuples in data tables expressing degrees to which tuples match queries.
  In this setting, we show that ranks assigned to tuples are insensitive to
  small changes, i.e., small changes in the input data do not yield large changes
  in the results of queries.


  \keywords{%
    declarative query languages,
    ordinal ranks,
    relational databases,
    residuated lattices}
\end{abstract}

\abovedisplayshortskip=6pt
\abovedisplayskip=6pt
\belowdisplayshortskip=6pt
\belowdisplayskip=6pt

\sloppy

\section{Introduction}\label{sec:intro}
Since its inception, the relational model of data introduced by E. Codd \cite{Co:Armodflsdb}
has been extensively studied by both computer scientists and database systems developers.
The model has become the standard theoretical model of relational data and
the formal foundation for relational database management systems. 
Various reasons for the success and strong position of Codd's model are analyzed 
in~\cite{Dat:DRM}, where the author emphasizes that the main
virtues of the model like logical and physical data independence, declarative
style of data retrieval (database querying), access flexibility and data
integrity are consequences of a close connection between the model and the
first-order predicate logic.

%

This paper is a continuation of our previous work~\cite{BeVy:Lffsbd,BeVy:Qssbd}
where we have introduced an extension of Codd's model in which tuples are
assigned ordinal ranks. The motivation for the model is that in many situations,
it is natural to consider not only the exact matches of queries in which a tuple
of values either \emph{does} or \emph{does not} match a query $Q$ but also
approximate matches where tuples match queries to degrees. The degrees of
approximate matches can usually be described verbally using linguistic
modifiers like ``not at all (matches)'' ``almost (matches)'',
``more or less (matches)'', ``fully (matches)'', etc.
From the user's point of view, each data table in our extended relational model
consists of (i) an ordinary data table whose meaning is the same as in the
Codd's model and (ii) ranks assigned to all tuples in the original data table.
This way, we come up with a notion of a ranked data table (shortly, an RDT).
The ranks in RDTs are interpreted as ``goodness of match'' 
and the interpretation
of RDTs is the same as in the Codd's model---they represent answers to queries
which are, in addition, equipped with priorities expressed by the ranks.
A user who looks at an answer to a query in our model is typically looking for
the best match possible represented by a tuple or tuples in the resulting
RDT with the highest ranks (i.e., highest priorities).

In order to have a suitable formalization of ranks and to perform operations
with ranked data tables, we have to choose a suitable structure for ranks. Since
ranks are meant to be compared by users, the set $L$ of all considered ranks
should be equipped with a partial order $\leq$,
i.e. $\langle L,\leq\rangle$ should be a poset.
Moreover, it is convenient to postulate that $\langle L,\leq\rangle$ is
a complete lattice~\cite{Bir:LT},
i.e., for each subset $A \subseteq L$, its least upper
bound (a supremum) and greatest lower bound (an infimum) exist. This way,
for any $A \subseteq L$, one can take the least rank in $L$ which represents
a higher priority (a better match) than all ranks from $A$. Such a rank is
then the supremum of $A$ (dually for the infimum). Since $\langle L,\leq\rangle$
is a complete lattice, it contains the least element denoted $0$ (no match at all)
and the greatest element denoted $1$ (full match).

The set $L$ of all ranks should also be equipped with additional operations
for aggregation of ranks. Indeed,
if tuple $t$ with rank $a$ is obtained as one of the results of subquery $Q_1$
and the same $t$ with another rank $b$ is obtained from answers to subquery
$Q_2$ then we might want to express the rank to which $t$ matches
a compound conjunctive query ``$Q_1$ \emph{and} $Q_2$''.
A natural way to do so is to take a suitable
binary operation $\otimes\!: L \times L \to L$ which acts as a conjunctor and
take $a \otimes b$ for the resulting rank. Obviously, not every binary operation
on $L$ represents a (reasonable) conjunctor, i.e. we may restrict the choices
only to particular binary operations that make ``good conjunctors''. There are
various ways to impose such restrictions. In our model,
we follow the approach of using residuated conjunctions that has proved to
be useful in logics based on residuated lattices~\cite{Bel:FRS,Got:Mfl,Haj:MFL}.
Namely, we assume that $\langle L,\otimes,1 \rangle$ is a commutative monoid
(i.e., $\otimes$ is associative, commutative, and neutral with respect to $1$)
and there is a binary operation $\rightarrow$ on $L$ such that
for all $a,b,c \in L$:
\begin{align}
  a \otimes b \leq c \quad \text{if and only if} \quad a \leq b \rightarrow c.
  \label{eqn:adj}
\end{align}
Operations $\otimes$ (a multiplication) and $\rightarrow$
(a residuum) satisfying \eqref{eqn:adj} are called \emph{adjoint operations.}
Altogether, the structure for ranks we use is
a \emph{complete residuated lattice}
$\mathbf{L}=\langle L,\wedge,\vee,\otimes,\rightarrow,0,1\rangle$, i.e.,
a complete lattice in which $\otimes$ and $\rightarrow$ are adjoint operations,
and $\wedge$ and $\vee$ denote the operations of infimum and supremum,
respectively. Considering $\mathbf{L}$ as a basic structure of ranks brings
several benefits. First, in multiple-valued logics and in particular fuzzy
logics~\cite{Got:Mfl,Haj:MFL}, residuated lattices are interpreted as structures
of truth degrees and the relationship \eqref{eqn:adj} between $\otimes$
(a fuzzy conjunction) and $\rightarrow$ (a fuzzy implication) is derived from
requirements on graded counterpart of the \emph{modus ponens} deduction rule
(currently, there are many strong-complete logics based on residuated lattices).

\begin{remark}
  The graded counterpart of \emph{modus ponens}~\cite{Haj:MFL,Pav:Ofl} can be seen
  as a generalized deduction rule saying ``from $\varphi$ valid (at least)
  to degree $a \in L$ and $\varphi \Rightarrow \psi$ valid (at least)
  to degre $b \in L$, infer $\psi$ valid (at least) to degree $a \otimes b$''.
  If if-part of \eqref{eqn:adj} ensures that the rule is sound while the
  only-if part ensures that it is as powerful as possible, i.e., $a \otimes b$
  is the highest degree to which we infer $\psi$ valid provided that 
  $\varphi$ valid at least to degree $a$ and $\varphi \Rightarrow \psi$ valid
  at least to degre $b \in L$. This relationship between $\rightarrow$
  (a truth function for logical connective imlication $\Rightarrow$) and
  $\otimes$ has been discovered in~\cite{Gog:Lic} and later used, e.g.,
  in~\cite{Ger:FL,Pav:Ofl}. Interestingly, \eqref{eqn:adj} together with the lattice
  ordering ensure enough properties of $\rightarrow$ and $\otimes$. For instance,
  $\rightarrow$ is antitone in the first argument and is monotone in the second
  one, condition $a \leq b$ if{}f $a \rightarrow b = 1$ holds for all
  $a,b \in L$, $a \rightarrow (b \rightarrow c)$ equals
  $(a \otimes b) \rightarrow c$ for all $a,b,c \in L$, etc.
  Since complete residuated lattices are in general weaker structures than Boolean
  algebras, not all laws satisfied by truth functions of the classic
  conjunction and implication are preserved by all complete residuated lattices.
  For instance, neither $a \otimes a = a$ (idempotency of $\otimes$) nor
  $(a \rightarrow 0) \rightarrow 0 = a$ (the law of double negation) nor
  $a \vee (a \rightarrow 0) = 1$ (the law of the excluded middle) hold in general.
  Nevertheless, complete residuated lattices are strong enough to provide a formal
  framework for relational analysis and similarity-based reasoning as it has been
  shown by previous results.
\end{remark}

Second, our extension of the Codd's model results from the model by replacing
the two-element Boolean algebra, which is the classic structure of truth values,
by a more general structure of truth values represented by a residuated lattice,
i.e. we make the following shift in (the semantics of) the underlying logic:
\begin{center}
  \emph{two-element Boolean algebra}
  \quad $\Longmapsto$ \quad
  \emph{a complete residuated lattice.}
\end{center}
Third, the original Codd's model is a special case of our model for $\mathbf{L}$
being the two-element Boolean algebra (only two borderline ranks $1$ and $0$ are
available). As a practical consequence, data tables in the Codd's model can be
seen as RDTs where all ranks are either equal to $1$ (full match) or $0$
(no match; tuples with $0$ rank are considered as not present in the result of
a query). Using residuated lattices as structures of truth degrees, we obtain
a generalization of Codd's model which is based on solid logical foundations and
has desirable properties. In addition, its relationship to residuated first-order
logics is the same as the relationship of the original Codd's model to the
classic first-order logic. The formalization we offer can further be used to
provide insight into several isolated approaches that have been
provided in the past, see e.g. \cite{BuPe:Fdne}, \cite{Fag:Cfi},
\cite{Liea:Rsql}, \cite{RaMa:Ffdljdfrds}, \cite{ShMe:Prfrdm},
\cite{Tak:Fdqltrct}, and a comparison paper~\cite{BeVy:Crmdfl}.

\begin{table}[t]
  \centering
  \caption{Houses for sale at \val{\$200,000} with square footage \val{1200}}
  \label{tab:sales_a}
  \small
  \setlength{\tabcolsep}{4pt}
  \renewcommand{\arraystretch}{0.9}
  \begin{tabular}{|r||lcrclc|}
    \hline &
    \atr{AGENT} & \atr{ID} & \atr{SQFT} &
    \atr{AGE} & \atr{LOCATION} & \atr{PRICE} \\
    \hline
    \rule{0pt}{8pt}%
    $0.93$ & \val{Brown} & \val{138} &
    \val{1185} & \val{48} & \val{Vestal} & \val{\$228,500} \\
    $0.89$ & \val{Clark} & \val{140} &
    \val{1120} & \val{30} & \val{Endicott} & \val{\$235,800} \\
    $0.86$ & \val{Brown} & \val{142} &
    \val{950} & \val{50} & \val{Binghamton} & \val{\$189,000} \\
    $0.85$ & \val{Brown} & \val{156} &
    \val{1300} & \val{85} & \val{Binghamton} & \val{\$248,600} \\
    $0.81$ & \val{Clark} & \val{158} &
    \val{1200} & \val{25} & \val{Vestal} & \val{\$293,500} \\
    $0.81$ & \val{Davis} & \val{189} &
    \val{1250} & \val{25} & \val{Binghamton} & \val{\$287,300} \\
    $0.75$ & \val{Davis} & \val{166} &
    \val{1040} & \val{50} & \val{Vestal} & \val{\$286,200} \\
    $0.37$ & \val{Davis} & \val{112} &
    \val{1890} & \val{30} & \val{Endicott} & \val{\$345,000} \\
    \hline
  \end{tabular}
\end{table}

A typical choice of $\mathbf{L}$ is a structure with $L=[0,1]$ (ranks are
taken from the real unit interval), $\wedge$ and $\vee$ being minimum and
maximum, $\otimes$ being a left-continuous (or a continuous) t-norm with
the corresponding $\rightarrow$, see~\cite{Bel:FRS,Got:Mfl,Haj:MFL}.
For example, an RDT with ranks coming from
such $\mathbf{L}$ is in Table~\ref{tab:sales_a}. It can be seen as a result of
similarity-based query ``show all houses which are sold for (approximately)
\val{\$200,000} and have (approximately) \val{1200} square feet''. The left-most
column contains ranks. The remaining part of the table is a data table in the
usual sense containing tuples of values. At this point, we do not explain in
detail how the particular ranks in Table~\ref{tab:sales_a} have been obtained
(this will be outlined in further sections). One way is by executing a
similarity-based query that uses additional information about similarity
(proximity) of domain values which is also described using degrees
from $\mathbf{L}$. Note that the concept of a similarity-based query appears
when human perception is involved in rating or comparing close values from
domains where not only the exact equalities (matches) are interesting.
For instance, a person 
searching in a database of houses is usually not
interested in houses sold for a particular exact price. Instead, the person
wishes to look at houses sold approximately at that price, including those which
are sold for other prices that are sufficiently close. While the ranks
constitute a ``visible'' part of any RDT, the similarities are not a direct
part of RDT and have to be specified for each domain independently. They can be
seen as an additional (background) information about domains which is supplied
by users of the database system.

Let us stress the meaning of ranks as priorities. As it is usual in fuzzy logics
in narrow sense, their meaning is primarily \emph{comparative},
cf.~\cite[p. 2]{Haj:MFL} and the comments on comparative meaning of truth degrees
therein. In our example, it means that tuple
$\langle \val{Clark}, \val{140}, \val{1120}, \val{30}, \val{Endicott},
\val{\$235,800}\rangle$ with rank $0.89$ is a better match than tuple
$\langle \val{Brown}, \val{142}, \val{950}, \val{50}, \val{Binghamton},
\val{\$189,000}\rangle$ whose rank $0.86$ is strictly smaller.
Thus, for end-users, the numerical values of ranks (if $L$ is a unit interval)
are not so important, the important thing is the relative ordering of tuples
given by the ranks.

Note that our model which provides theoretical foundations for similarity-based
databases~\cite{BeVy:Lffsbd,BeVy:Qssbd} should not be confused with models
for probabilistic databases~\cite{SiOlReKo:PD} which have recently been studied,
e.g. in \cite{CaPi:Ttopd,DaSu:Eqeopd,DaSu:Mopdfac,ImLi:Iiird,Ko:Oqafpd,OlHuKo:Accipd},
see also~\cite{DaReSu:Pdditd} for a survey. In particular, numerical ranks used
in our model (if $L=[0,1]$) cannot be interpreted as probabilities, confidence
degrees of belief degrees as in case of probabilistic databases where ranks play
such roles. In probabilistic databases, the tuples (i.e., the data itself) are
uncertain and the ranks express probabilities that tuples appear in data tables.
Consequently, a probabilistic database is formalized by a discrete probability
space over the possible contents of the database~\cite{DaReSu:Pdditd}.
Nevertheless, the underlying logic of the models is the classical two-valued
first-order logic---only yes/no matches are allowed (with uncertain outcome).
In our case, the situation is quite different. The data (represented by tuples)
is absolutely certain but the tuples are allowed to match queries to degrees.
This, translated in terms of logic, means that formulas (encoding queries) are
allowed to be evaluated to truth degrees other than $0$ and $1$. Therefore, the
underlying logic in our model is not the classic two-element Boolean logic as
we have argued hereinbefore.

In~\cite{Abea:Ldrsa}, a report written by leading authorities in database
systems, the authors say that the current database management systems 
have no facilities for
either approximate data or imprecise queries. According to this report, the
management of uncertainty and imprecision is one of the six currently most
important research directions in database systems. Nowadays, probabilistic
databases (dealing with approximate data) are extensively studied.
On the contrary, it seems that similarity-based databases (dealing with imprecise
queries) have not yet been paid full attention. This paper is a contribution to
theoretical foundations of similarity-based databases.

\section{Problem Setting}\label{sec:problem}
The issue we address in this paper is the following. In our model, we can get
two or more RDTs (as results of queries) which are not exactly the same but
which are perceived (by users) as being similar. For instance, one can obtain
two RDTs containing the same tuples with numerical values of ranks that are
almost the same. A question is whether such similar RDTs, when used in
subsequent queries, yield similar results. In this paper, we present a
preliminary study of the phenomenon of similarity of RDTs and its relationship
to the similarity of query results obtained by applying queries to similar input
data tables. We present basic notions and results providing formulas for
computing estimations of similarity degrees. The observations we present
provide a formal justification for the phenomenon discussed in the previous
section---slight changes in ranks do not have a large impact on the results
of (complex) queries. The results are obtained for any complete residuated
lattice taken as the structure of ranks (truth degrees). Note that the basic
query systems in our model are (extensions of) domain relational
calculus~\cite{BeVy:Qssbd,Mai:TRD} and relational
algebra~\cite{BeVy:Lffsbd,Mai:TRD}. We formulate the results in terms of
operations of the relational algebra but due to its equivalence with the domain
relational calculus~\cite{BeVy:Qssbd}, the results pertain to both the query
systems. Thus, based on the domain relational calculus, one may design
a declarative query language preserving similarity in which execution of queries
is based on transformations to expressions of relational algebra in a similar
way as in the classic case~\cite{Mai:TRD}.

The rest of the paper is organized as follows. Section~\ref{sec:prelim}
presents a short survey of notions. Section~\ref{sec:simil} contains results
on sensitivity analysis, an illustrative example, and a short outline of future
research. Because of the limited scope of the paper, proofs are sketched or
omitted.

\section{Preliminaries}\label{sec:prelim}
In this section, we recall basic notions of RDTs and relational operations we
need to provide insight into the sensitivity issues of RDTs in
Section~\ref{sec:simil}. Details can be
found in~\cite{Bel:FRS,BeVy:Lffsbd,BeVy:Crmdfl}.
In the rest of the paper, $\mathbf{L}$
always refers to a complete residuated lattice
$\mathbf{L}=\langle L,\wedge,\vee,\otimes,\rightarrow,0,1\rangle$,
see Section~\ref{sec:intro}. 

\subsection{Basic Structures}
Given $\mathbf{L}$, we make use of the following notions: An $\mathbf{L}$-set $A$
in universe $U$ is a map $A\!:U \to L$, $A(u)$ being interpreted as ``the degree
to which $u$ belongs to $A$''. If $\mathbf{L}$ is the two-element Boolean algebra,
then $A\!:U \to L$ is an indicator function of a classic subset of $U$,
$A(u) = 1$ ($A(u) = 0$) meaning that $u$ belongs (does not belong) to that subset.
In our approach, we tacitly identify sets with their indicator functions.
In a similar way, a binary $\mathbf{L}$-relation $B$ on $U$ is a map
$B\!: U \times U \to L$, $B(u_1,u_2)$ interpreted as ``the degree to which
$u_1$ and $u_2$ are related according to $B$''. Hence, $B$ is an $\mathbf{L}$-set
in universe $U \times U$.

\subsection{Ranked Data Tables over Domains with Similarities}
We denote by $Y$ a set of \emph{attributes}, any subset $R \subseteq Y$ is
called a \emph{relation scheme.} For each attribute $y \in Y$ we consider
its \emph{domain $D_y$}. In addition, each $D_y$ is equipped
with a binary $\mathbf{L}$-relation $\approx_y$ on $D_y$ satisfying
reflexivity ($u \approx_y u = 1$) and
symmetry $u \approx_y v = v \approx_y u$ (for all $u,v \in D_y$).
Each binary $\mathbf{L}$-relation $\approx_y$ on $D_y$ satisfying
(i) and (ii) shall be called a \emph{similarity.}
Pair $\langle D_y,\approx_y\rangle$ is called a \emph{domain with similarity.}

Tuples contained in data tables will be considered as usual, i.e., as elements
of Cartesian products of domains. Recall that a Cartesian product
$\prod_{i \in I}D_i$ of an $I$-indexed system $\{D_i \,|\, i \in I\}$ of
sets $D_i$ ($i \in I$) is a set of all maps $t\!: I \to \bigcup_{i \in I}D_i$
such that $t(i) \in D_i$ holds for each $i \in I$. Under this notation,
a~\emph{tuple} over $R \subseteq Y$ is any element from $\prod_{y \in R}D_y$.
For brevity, $\prod_{y \in R}D_y$ is denoted by $\mathrm{Tupl}(R)$.
Following the example in Table~\ref{tab:sales_a}, tuple
$\langle \val{Brown}, \val{142}, \val{950}, \val{50}, \val{Binghamton},
\val{\$189,000}\rangle$ is a map $r \in \mathrm{Tupl}(R)$ for
$R = \{\atr{AGENT},\atr{ID},\ldots,\atr{PRICE}\}$ such that
$r(\atr{AGENT})=\val{Brown}$, $r(\atr{ID})=\val{142}$, etc.

A \emph{ranked data table} on $R \subseteq Y$
over $\{\langle D_y,\approx_y\rangle \,|\, y \in R\}$ (shortly, an RDT)
is any (finite) $\mathbf{L}$-set $\mathcal{D}$ in $\mathrm{Tupl}(R)$.
The degree $\mathcal{D}(r)$ to which $r$ belongs to $\mathcal{D}$ is
called a \emph{rank} of tuple $r$ in $\mathcal{D}$.
According to its definition, if $\mathcal{D}$ is an RDT
on $R$ over $\{\langle D_y,\approx_y\rangle \,|\, y \in R\}$ then
$\mathcal{D}$ is a map $\mathcal{D}\!: \mathrm{Tupl}(R) \to L$.
Note that $\mathcal{D}$ is an $n$-ary $\mathbf{L}$-relation between
domains $D_y$ ($y \in Y$) since $\mathcal{D}$ is a map from
$\textstyle\prod_{y \in R}D_y$ to $L$. In our example,
$\mathcal{D}(r) = 0.86$ for $r$ being the tuple with 
$r(\atr{ID})=\val{142}$.%

\subsection{Relational Operations with RDTs}
Relational operations we consider in this paper are the following:
For RDTs $\mathcal{D}_1$ and $\mathcal{D}_2$ on $T$, we put
$(\mathcal{D}_1 \cup \mathcal{D}_2)(t) =
\mathcal{D}_1(t) \vee \mathcal{D}_2(t)$
and $(\mathcal{D}_1 \cap \mathcal{D}_2)(t) =
\mathcal{D}_1(t) \wedge \mathcal{D}_2(t)$
for each $t \in \mathrm{Tupl}(T)$;
$\mathcal{D}_1 \cup \mathcal{D}_2$ and $\mathcal{D}_1 \cap \mathcal{D}_2$
are called the \emph{union} and the \emph{$\wedge$-intersection} of
$\mathcal{D}_1$ and $\mathcal{D}_2$, respectively.
Analogously, one can define an \emph{$\otimes$-intersection}
$\mathcal{D}_1 \otimes \mathcal{D}_2$.
Hence, $\cup$, $\cap$, and $\otimes$ are defined componentwise
based on the operations of the complete residuated lattice $\mathbf{L}$.

Moreover, our model admits new operations that are trivial in the classic
model. For instance, for $a \in L$, we introduce an
\emph{$a$-shift $a{\rightarrow}\mathcal{D}$} of $\mathcal{D}$ by
$(a{\rightarrow}\mathcal{D})(t) = a \rightarrow \mathcal{D}(t)$
for all $t \in \mathrm{Tupl}(T)$.

\begin{remark}
  Note that if $\mathbf{L}$ is the two-element Boolean algebra then $a$-shift is
  a trivial operation since $1 \rightarrow \mathcal{D} = \mathcal{D}$ and 
  $0 \rightarrow \mathcal{D}$ produces a possibly infinite table
  containing all tuples from $\mathrm{Tupl}(T)$. In our model,
  an $a$-shift has the following meaning:
  If $\mathcal{D}$ is a result of query $Q$ then
  $(a{\rightarrow}\mathcal{D})(t)$ is a ``degree to which $t$ matches query
  $Q$ at least to degree $a$''. This follows from properties of residuum,
  see~\cite{Bel:FRS,Haj:MFL}. Hence, $a$-shifts allow us to emphasize results
  that match queries at least to a prescribed degree $a$.
\end{remark}

The remaining relational operations we consider represent counterparts of
projection, selection, and join in our model.
If $\mathcal{D}$ is an RDT on $T$,
the \emph{projection $\project{R}{\mathcal{D}}$} of $\mathcal{D}$
onto $R \subseteq T$ is defined by
\begin{align*}
  (\project{R}{\mathcal{D}})(r) =
  \textstyle\bigvee_{\!s \in \mathrm{Tupl}(T \setminus R)}\mathcal{D}(rs),
\end{align*}
for each $r \in \mathrm{Tupl}(R)$. In our example, the result of
$\project{\{\atr{LOCATION}\}}{\mathcal{D}}$ is a ranked data table with
single column such that 
$\project{\{\atr{LOCATION}\}}{\mathcal{D}}(\langle\val{Binghamton}\rangle) = 0.86$,
$\project{\{\atr{LOCATION}\}}{\mathcal{D}}(\langle\val{Vestal}\rangle) = 0.93$,
and
$\project{\{\atr{LOCATION}\}}{\mathcal{D}}(\langle\val{Endicott}\rangle) = 0.89$.

A similarity-based selection is a counterpart
to ordinary selection which selects from a data table all tuples which
approximately match a given condition: Let $\mathcal{D}$ be an RDT on $T$
and let $y \in T$ and $d \in D_y$. Then, a
\emph{similarity-based selection $\select{y \approx d}{\mathcal{D}}$}
of tuples in $\mathcal{D}$ matching $y \approx d$ is defined by
\begin{align*}
  \bigl(\select{y \approx d}{\mathcal{D}}\bigr)(t) =
  \mathcal{D}(t) \otimes t(y) \mathop{\approx_y} d.
\end{align*}
Considering $\mathcal{D}$ as a result of query $Q$, the rank of $t$
in $\select{y \approx d}{\mathcal{D}}$ can be interpreted as a degree to which 
``$t$ matches the query $Q$ and the $y$-value of $t$ is similar to $d$''.
In particular, an interesting case is $\select{p \approx q}{\mathcal{D}}$ where
$p$ and $q$ are both attributes with a common domain with similarity.

Similarity-based joins are considered as derived operations based on Cartrsian
products and similarity-based selections.
For $r \in \mathrm{Tupl}(R)$ and $s \in \mathrm{Tupl}(S)$ such that
$R \cap S = \emptyset$, we define a concatenation $rs \in \mathrm{Tupl}(R \cup S)$
of tuples $r$ and $s$ so that $(rs)(y) = r(y)$ for $y \in R$
and $(rs)(y) = s(y)$ for $y \in S$.
For RDTs $\mathcal{D}_1$ and $\mathcal{D}_2$ on disjoint relation schemes
$S$ and $T$ we define a RDT  $\mathcal{D}_1 \times \mathcal{D}_2$ on $S \cup T$,
called a \emph{Cartesian product} of $\mathcal{D}_1$ and $\mathcal{D}_2$, by
$(\mathcal{D}_1 \times \mathcal{D}_2)(st) =
\mathcal{D}_1(s) \otimes \mathcal{D}_2(t)$.
Using Cartesian products and similarity-based selections,
we can introduce \emph{similarity-based $\theta$-joins} such as 
$\join{p \approx q}{\mathcal{D}_1}{\mathcal{D}_2} =
\select{p \approx q}{\mathcal{D}_1 \times \mathcal{D}_2}$.
Various other types of similarity-based joins can be introduced in our model,
see~\cite{BeVy:Qssbd}.


\section{Estimations of Sensitivity of Query Results}\label{sec:simil}

\subsection{Rank-Based Similarity of Query Results}\label{sec:rnk}
We now introduce the notion of similarity of RDTs which is based on the idea
that RDTs $\mathcal{D}_1$ and $\mathcal{D}_2$ (on the same relation scheme) are
similar if{}f for each tuple $t$, ranks $\mathcal{D}_1(t)$ and $\mathcal{D}_2(t)$
are similar (degrees from $\mathbf{L}$). Similarity of ranks can be expressed by
biresiduum $\leftrightarrow$ (a fuzzy equivalence~\cite{Bel:FRS,Got:Mfl,Haj:MFL})
which is a derived operation of $\mathbf{L}$ such that
$a \leftrightarrow b = (a \rightarrow b) \wedge (b \rightarrow a)$.
Since we are interested in similarity of $\mathcal{D}_1(t)$ and $\mathcal{D}_2(t)$
for all possible tuples $t$, it is straightforward to define the similarity
$E(\mathcal{D}_1,\mathcal{D}_2)$ of $\mathcal{D}_1$ and $\mathcal{D}_2$
by an infimum which goes over all tuples:
\begin{align}
  E(\mathcal{D}_1,\mathcal{D}_2) &=
  \textstyle\bigwedge\nolimits_{t \in \mathrm{Tupl}(T)}
  \bigl(\mathcal{D}_1(t) \leftrightarrow \mathcal{D}_2(t)\bigr).
  \label{def:E}
\end{align}
An alternative (but equivalent) way is the following: we first formalize a degree
$S(\mathcal{D}_1,\mathcal{D}_2)$ to which $\mathcal{D}_1$ is included
in $\mathcal{D}_2$. We can say that $\mathcal{D}_1$ is fully included in
$\mathcal{D}_2$ if{}f, for each tuple $t$, the rank $\mathcal{D}_2(t)$ is
at least as high as the rank $\mathcal{D}_1(t)$. Notice that in the classic
(two-values) case, this is exactly how one defines the ordinary subsethood
relation ``$\subseteq$''. Considering general degrees of inclusion (subsethood),
a degree $S(\mathcal{D}_1,\mathcal{D}_2)$ to which $\mathcal{D}_1$ is
included in $\mathcal{D}_2$ can be defined as follows:
\begin{align}
  S(\mathcal{D}_1,\mathcal{D}_2) &=
  \textstyle\bigwedge\nolimits_{t \in \mathrm{Tupl}(T)}
  \bigl(\mathcal{D}_1(t) \rightarrow \mathcal{D}_2(t)\bigr).
  \label{def:S}
\end{align}
It is easy to prove~\cite{Bel:FRS} that \eqref{def:E} and \eqref{def:S} satisfy:
\begin{align}
  E(\mathcal{D}_1,\mathcal{D}_2) &=
  S(\mathcal{D}_1,\mathcal{D}_2) \wedge
  S(\mathcal{D}_2,\mathcal{D}_1).
  \label{eqn:ES}
\end{align}
Note that $E$ and $S$ defined by \eqref{def:E} and \eqref{def:S} are
known as degrees of similarity and subsethood from general fuzzy
relational systems~\cite{Bel:FRS} (in this case, the fuzzy relations are RDTs).

The following assertion shows that $\cup$, $\cap$, $\otimes$, and $a$-shifts
preserve subsethood degrees given by~\eqref{def:S}. In words, the degree to
which $\mathcal{D}_1 \cup \mathcal{D}_2$ is included
in $\mathcal{D}'_1 \cup \mathcal{D}'_2$ is at least as high as the degree
to which $\mathcal{D}_1$ is included in $\mathcal{D}'_1$ and 
$\mathcal{D}_2$ is included in $\mathcal{D}'_2$. A similar verbal description
can be made for the other operations.

\begin{theorem}\label{th:inc_op}
  For any $\mathcal{D}_1$, $\mathcal{D}'_1$,
  $\mathcal{D}_2$, and $\mathcal{D}'_2$ on relation scheme $T$,
  \begin{align}
    S(\mathcal{D}_1,\mathcal{D}'_1) \wedge S(\mathcal{D}_2,\mathcal{D}'_2)
    &\leq
    S(\mathcal{D}_1 \cup \mathcal{D}_2,
    \mathcal{D}'_1 \cup \mathcal{D}'_2),
    \label{eqn:inc_union} \\
    S(\mathcal{D}_1,\mathcal{D}'_1) \wedge S(\mathcal{D}_2,\mathcal{D}'_2)
    &\leq
    S(\mathcal{D}_1 \cap \mathcal{D}_2,
    \mathcal{D}'_1 \cap \mathcal{D}'_2),
    \label{eqn:inc_intersection} \\
    S(\mathcal{D}_1,\mathcal{D}'_1) \otimes S(\mathcal{D}_2,\mathcal{D}'_2)
    &\leq
    S(\mathcal{D}_1 \otimes \mathcal{D}_2,
    \mathcal{D}'_1 \otimes \mathcal{D}'_2),
    \label{eqn:inc_otimes} \\
    S(\mathcal{D}_1, \mathcal{D}_2)
    &\leq 
    S(a\rightarrow \mathcal{D}_1, a \rightarrow \mathcal{D}_2).
    \label{eqn:inc_a_shift}
  \end{align}
\end{theorem}
\begin{proof}[sketch]
  \eqref{eqn:inc_union}:
  Using adjointness, it suffices to check that
  $\bigl(S(\mathcal{D}_1,\mathcal{D}'_1) \wedge
  S(\mathcal{D}_2,\mathcal{D}'_2)\bigr) \otimes
  (\mathcal{D}_1 \cup \mathcal{D}_2)(t) \leq
  (\mathcal{D}'_1 \cup \mathcal{D}'_2)(t)$ holds
  true for any $t \in \mathrm{Tupl}(T)$. Using 
  \eqref{def:S}, the monotony of $\otimes$ and $\wedge$ yields
  $\bigl(S(\mathcal{D}_1,\mathcal{D}'_1) \wedge
    S(\mathcal{D}_2,\mathcal{D}'_2)\bigr) \otimes
    (\mathcal{D}_1 \cup \mathcal{D}_2)(t) \leq
    \bigl((\mathcal{D}_1(t) \rightarrow \mathcal{D}'_1(t))
    \wedge (\mathcal{D}_2(t) \rightarrow \mathcal{D}'_2(t))\bigr)
    \otimes (\mathcal{D}_1(t) \vee \mathcal{D}_2(t))$.
  Applying $a \otimes (b \vee c) = (a \otimes b) \vee (a \otimes c)$
  to the latter expression, we get
$    \bigl((\mathcal{D}_1(t) \rightarrow \mathcal{D}'_1(t)) \wedge
    (\mathcal{D}_2(t) \rightarrow \mathcal{D}'_2(t)) \bigr) \otimes
    (\mathcal{D}_1(t) \vee \mathcal{D}_2(t))
    \leq
    \bigl((\mathcal{D}_1(t) \rightarrow \mathcal{D}'_1(t)) \otimes
    \mathcal{D}_1(t)\bigr) \vee \bigl((\mathcal{D}_2(t) \rightarrow
    \mathcal{D}'_2(t)) \otimes \mathcal{D}_2(t)\bigr)$.
  Using $a \otimes (a \rightarrow b) \leq b$ twice, it follows that
    $\bigl((\mathcal{D}_1(t) \rightarrow \mathcal{D}'_1(t)) \otimes
    \mathcal{D}_1(t)\bigr) \vee \bigl((\mathcal{D}_2(t) \rightarrow
    \mathcal{D}'_2(t)) \otimes \mathcal{D}_2(t)\bigr) \leq \mathcal{D}'_1(t)
    \vee \mathcal{D}'_2(t)$.
  Putting previous inequalities together,
  $\bigl(S(\mathcal{D}_1,\mathcal{D}'_1) \wedge
  S(\mathcal{D}_2,\mathcal{D}'_2)\bigr) \otimes
  (\mathcal{D}_1 \cup \mathcal{D}_2)(t) \leq 
  (\mathcal{D}'_1 \cup \mathcal{D}'_2)(t)$ which proves \eqref{eqn:inc_union}.
  \eqref{eqn:inc_intersection} can be proved analogously
  as \eqref{eqn:inc_union};
  \eqref{eqn:inc_otimes} can be proved analogously as
  \eqref{eqn:inc_intersection} using monotony of $\otimes$;
  \eqref{eqn:inc_a_shift} follows from the fact that
  $a\rightarrow b \leq (c\rightarrow a)\rightarrow (c\rightarrow b)$.
  \qed

\end{proof}

Using~\eqref{eqn:ES}, we have the following consequence of
Theorem~\ref{th:inc_op}:

\begin{corollary}\label{thm:sim_setop}
  For $\lozenge$ being $\cap$ and $\cup$, we have:
  \begin{align}
    E(\mathcal{D}_1,\mathcal{D}'_1) \wedge E(\mathcal{D}_2,\mathcal{D}'_2)
    &\leq
    E(\mathcal{D}_1 \mathop{\lozenge} \mathcal{D}_2,
    \mathcal{D}'_1 \mathop{\lozenge} \mathcal{D}'_2).
    \label{eqn:sim_wedge_op} \\
    E(\mathcal{D}_1,\mathcal{D}'_1) \otimes E(\mathcal{D}_2,\mathcal{D}'_2)
    &\leq
    E(\mathcal{D}_1 \otimes \mathcal{D}_2,
    \mathcal{D}'_1 \otimes \mathcal{D}'_2).
    \label{eqn:sim_otimes_op} \\
    E(\mathcal{D}_1,\mathcal{D}_2)
    &\leq
    E(a \rightarrow \mathcal{D}_1, a \rightarrow \mathcal{D}_2).
   \label{eqn:sim_ashift_op}
 \end{align}
\end{corollary}
\begin{proof}[sketch]
  For $\lozenge$ being $\cap$, \eqref{eqn:inc_intersection}
  applied twice yields:
  $S(\mathcal{D}_1,\mathcal{D}'_1) \wedge S(\mathcal{D}_2,\mathcal{D}'_2)
  \leq
  S(\mathcal{D}_1 \cap \mathcal{D}_2,
  \mathcal{D}'_1 \cap \mathcal{D}'_2)$
  and
  $S(\mathcal{D}'_1,\mathcal{D}_1) \wedge S(\mathcal{D}'_2,\mathcal{D}_2)
  \leq
  S(\mathcal{D}'_1 \cap \mathcal{D}'_2,
  \mathcal{D}_1 \cap \mathcal{D}_2)$.
  Hence, \eqref{eqn:sim_wedge_op} for $\cap$ follows 
  using \eqref{def:E}. The rest is analogous.
  \qed
\end{proof}

Using the idea in the proof of Corollary~\ref{thm:sim_setop}, in order to prove
that operation $O$ preserves similarity, it suffices to check that $O$
preserves (graded) subsethood. Thus, from now on, we shall only investigate
whether operations preserve subsethood. In case of Cartesian products,
we have:

\begin{theorem}\label{thm:sim_cart*}
  Let $\mathcal{D}_1$ and $\mathcal{D}'_1$ be RDTs on relation scheme $S$ and
  let $\mathcal{D}_2$ and $\mathcal{D}'_2$ be RDTs on relation scheme $T$
  such that $S \cap T = \emptyset$. Then,
  \begin{align}
    S(\mathcal{D}_1,\mathcal{D}'_1) \otimes S(\mathcal{D}_2,\mathcal{D}'_2)
    &\leq
    S(\mathcal{D}_1 \times \mathcal{D}_2,
    \mathcal{D}'_1 \times \mathcal{D}'_2),
    \label{eqn:inc_cart*} 
  \end{align}
\end{theorem}
\begin{proof}[sketch]
  The proof is analogous to that of \eqref{eqn:inc_otimes}.
  \qed
\end{proof}

The following assertion shows that projection and similarity-based selection
preserve subsethood degrees (and therefore similarities) of RDTs:

\begin{theorem}\label{thm:sim_proj_select}
  Let $\mathcal{D}$ and $\mathcal{D}'$ be RDTs on relation scheme $T$ and
  let $y \in T$, $d \in D_y$, and $R \subseteq T$. Then,
  \begin{align}
    S(\mathcal{D},\mathcal{D}') &\leq
    S(\project{R}{\mathcal{D}}, \project{R}{\mathcal{D}'}),
    \label{eqn:sim_proj} \\
    S(\mathcal{D},\mathcal{D}') &\leq
    S(\select{y \approx d}{\mathcal{D}}, \select{y \approx d}{\mathcal{D}'}).
    \label{eqn:sim_select}
  \end{align}
\end{theorem}
\begin{proof}[sketch]
  In oder to prove \eqref{eqn:sim_proj}, we check
  $S(\mathcal{D},\mathcal{D}') \otimes (\project{R}{\mathcal{D}})(r) \leq
  (\project{R}{\mathcal{D}'})(r)$ for any $r \in \mathrm{Tupl}(R)$.
  It means showing that
  \begin{align*}
    S(\mathcal{D},\mathcal{D}') \otimes
    \textstyle\bigvee_{\!s \in \mathrm{Tupl}(T \setminus R)}\mathcal{D}(rs)
    \leq (\project{R}{\mathcal{D}'})(r).
  \end{align*}
  Thus, is suffices to prove
  $S(\mathcal{D},\mathcal{D}') \otimes \mathcal{D}(rs) \leq 
  (\project{R}{\mathcal{D}'})(r)$ for all
  $s \in \mathrm{Tupl}(T \setminus R)$.
  Using monotony of $\otimes$, we get
  $S(\mathcal{D},\mathcal{D}') \otimes \mathcal{D}(rs) \leq
  (\mathcal{D}(rs) \rightarrow \mathcal{D}'(rs)) \otimes \mathcal{D}(rs)
  \leq \mathcal{D}'(rs)$,
  because $rs \in \mathrm{Tupl}(T)$. Therefore,
  $S(\mathcal{D},\mathcal{D}') \otimes \mathcal{D}(rs)
  \leq \mathcal{D}'(rs) \leq
  \textstyle\bigvee_{\!s \in \mathrm{Tupl}(T \setminus R)}\mathcal{D}'(rs) =
  (\project{R}{\mathcal{D}'})(r)$,
  which proves the first claim of \eqref{eqn:sim_proj}.
  In case of \eqref{eqn:sim_select}, we proceed analogously.
  \qed
\end{proof}

Theorem~\ref{thm:sim_cart*} and Theorem~\ref{thm:sim_proj_select}
used together yield

\begin{corollary}
  Let $\mathcal{D}_1$ and $\mathcal{D}'_1$ be RDTs on relation scheme $S$ and
  let $\mathcal{D}_2$ and $\mathcal{D}'_2$ be RDTs on relation scheme $T$
  such that $S \cap T = \emptyset$. Then,
  \begin{align}
    S(\mathcal{D}_1,\mathcal{D}'_1) \otimes S(\mathcal{D}_2,\mathcal{D}'_2) &\leq
    S(\join{p \approx q}{\mathcal{D}_1}{\mathcal{D}_2},
    \join{p \approx q}{\mathcal{D}'_1}{\mathcal{D}'_2}).
    \label{eqn:inc_join*} 
  \end{align}
  for any $p \in S$ and $q \in T$ having the same domain with similarity.
  \qed
\end{corollary}

As a result, we have shown that important relational operations
in our model (including similarity-based joins)
preserve similarity defined by~\eqref{def:E}. Thus, we have provided
a formal justification for the (intuitively expected but nontrivial)
fact that similar input data yield similar results of queries.

\begin{table}[t]
  \centering
  \caption{Alternative ranks for houses for sale from Table~\ref{tab:sales_a}}
  \label{tab:sales_b}
  \setlength{\tabcolsep}{4pt}
  \renewcommand{\arraystretch}{0.9}
  \begin{tabular}{|r||lcrclc|}
    \hline &
    \atr{AGENT} & \atr{ID} & \atr{SQFT} &
    \atr{AGE} & \atr{LOCATION} & \atr{PRICE} \\
    \hline
    \rule{0pt}{8pt}%
    $0.93$ & \val{Brown} & \val{138} &
    \val{1185} & \val{48} & \val{Vestal} & \val{\$228,500} \\
    $0.91$ & \val{Clark} & \val{140} &
    \val{1120} & \val{30} & \val{Endicott} & \val{\$235,800} \\
    $0.87$ & \val{Brown} & \val{156} &
    \val{1300} & \val{85} & \val{Binghamton} & \val{\$248,600} \\
    $0.85$ & \val{Brown} & \val{142} &
    \val{950} & \val{50} & \val{Binghamton} & \val{\$189,000} \\
    $0.82$ & \val{Davis} & \val{189} &
    \val{1250} & \val{25} & \val{Binghamton} & \val{\$287,300} \\
    $0.79$ & \val{Clark} & \val{158} &
    \val{1200} & \val{25} & \val{Vestal} & \val{\$293,500} \\
    $0.75$ & \val{Davis} & \val{166} &
    \val{1040} & \val{50} & \val{Vestal} & \val{\$286,200} \\
    $0.37$ & \val{Davis} & \val{112} &
    \val{1890} & \val{30} & \val{Endicott} & \val{\$345,000} \\
    \hline
  \end{tabular}
\end{table}

\begin{remark}
  In this paper, we have restricted ourselves only to a fragment of relational
  operations in our model. In~\cite{BeVy:Qssbd}, we have shown that in order to
  have a relational algebra whose expressive power is the same as the expressive
  power of the domain relational calculus, we have to consider additional
  operations of \emph{residuum} (defined componentwise using $\rightarrow$)
  and \emph{division.} Nevertheless, these two additional operations preserve $E$
  as well---it can be shown using similar arguments as in the proof of
  Theorem~\ref{th:inc_op}. As a consequence, the similarity is preserved by all
  queries that can be formulated in DRC~\cite{BeVy:Qssbd}.
\end{remark}

\subsection{Illustrative Example}
Consider again the RDT from Table~\ref{tab:sales_a}. The RDT can be seen as
a result of querying a database of houses for sale where one wants to find
a house which is sold for (approximately) \val{\$200,000} and has
(approximately) \val{1200} square feet. The attributes in the RDT are:
real estate agent name (\atr{agent}), house ID (\atr{id}), square footage
(\atr{sqft}), house age (\atr{age}), house location (\atr{location}), and
house price (\atr{price}). In this example, the complete residuated
lattice $\mathbf{L}=\langle L,\wedge,\vee,\otimes,\rightarrow,0,1\rangle$
serving as the structure of ranks will be the so-called
\L ukasiewicz algebra~\cite{Bel:FRS,Got:Mfl,Haj:MFL}.
That is, $L=[0,1]$, $\wedge$ and $\vee$ are minimum and maximum, respectively,
and the multiplication and residuum are defined as follows:
$a \otimes b = \max(a + b - 1, 0)$ and $a \rightarrow b = \min(1 - a + b, 1)$
for all $a,b \in L$.

Intuitively, it is natural to consider similarity of
values in domains of \atr{sqft}, \atr{age}, \atr{location}, and \atr{price}.
For instance, similarity of prices can be defined by
$p_1 \mathop{\approx_{\atr{PRICE}}} p_2 = s(|p_2 - p_1|)$ using an antitone
scaling function $s\!: [0,\infty) \to [0,1]$ with $s(0) = 1$
(i.e., identical prices are fully similar). Analogously, a similarity
of locations can be defined based on their geographical distance and/or
based on their evaluation (safety, school districts,\,\ldots) by an expert.
In contrast, there is no need to have similarities for \atr{id} and \atr{agents}
because end-users do not look for houses based on (similarity of) their
(internal) IDs which are kept as keys merely because of performance reasons.
Obviously, there may be various reasonable similarity relations defined for
the above-mentioned domains and their careful choice is an important task.
In this paper, we neither explain nor recommend particular ways to do so because
(i) we try to keep a general view of the problem and (ii) similarities on
domains are purpose and user dependent.

Consider now the RDT in Table~\ref{tab:sales_b} defined over the same relation
scheme as the RDT in Table~\ref{tab:sales_a}. These two RDTs can be seen as two
(slightly different) answers to the same query (when e.g., the domain
similarities have been slightly changed) or answers to a modified
query (e.g., ``show all houses which are sold for (approximately)
\val{\$210,000} and\,\ldots''). The similarity of both the RDTs
given by \eqref{def:E} is $0.98$ (very high). The results in the previous
section say that if we perform any (arbitrarily complex) query (using the
relational operations we consider in this paper) with Table~\ref{tab:sales_b}
instead of Table~\ref{tab:sales_a}, the results will be similar at least
to degree $0.98$.

\begin{table}
  \centering
  \caption{Join of Table~\ref{tab:sales_a} and the table of customers}
  \label{tab:join}
  \setlength{\tabcolsep}{4pt}
  \renewcommand{\arraystretch}{0.9}
  \begin{tabular}{|r||ccccc|}
    \hline &
    \atr{AGENT} & \atr{ID} & \atr{PRICE} & \atr{NAME} & \atr{BUDGET} \\
    \hline
    \rule{0pt}{8pt}%
    $0.91$ & \val{Brown} & \val{138} & \val{\$228,500} &
    \val{Grant} & \val{\$240,000} \\
    $0.89$ & \val{Brown} & \val{138} & \val{\$228,500} &
    \val{Evans} & \val{\$250,000} \\
    $0.89$ & \val{Brown} & \val{138} & \val{\$228,500} &
    \val{Finch} & \val{\$210,000} \\
    $0.88$ & \val{Clark} & \val{140} & \val{\$235,800} &
    \val{Grant} & \val{\$240,000} \\
    $0.86$ & \val{Clark} & \val{140} & \val{\$235,800} &
    \val{Evans} & \val{\$250,000} \\
    $0.84$ & \val{Brown} & \val{156} & \val{\$248,600} &
    \val{Evans} & \val{\$250,000} \\[-4pt]
    \multicolumn{1}{c}{$\vdots$} &
    \multicolumn{1}{c}{$\vdots$} &
    \multicolumn{1}{c}{$\vdots$} &
    \multicolumn{1}{c}{$\vdots$} &
    \multicolumn{1}{c}{$\vdots$} &
    \multicolumn{1}{c}{$\vdots$} 
    \\
    $0.16$ & \val{Davis} & \val{112} & \val{\$345,000} &
    \val{Grant} & \val{\$240,000} \\
    $0.10$ & \val{Davis} & \val{112} & \val{\$345,000} &
    \val{Finch} & \val{\$210,000} \\
    \hline
  \end{tabular}
\end{table}

For illustration, consider an additional RDT of customers over relation
scheme containing two attributes: \atr{name} (customer name) and \atr{budget}
(price the customer is willing to pay for a house). In particular, let
$\langle \val{Evans}, \val{\$250,000}\rangle$,
$\langle \val{Finch}, \val{\$210,000}\rangle$, and
$\langle \val{Grant}, \val{\$240,000}\rangle$ be the only tuples
in the RDT (all with ranks $1$). The answer to the following query
\begin{align*}
  \project{\{\atr{AGENT},\atr{ID},\atr{PRICE},\atr{NAME},\atr{BUDGET}\}}
  {\join{\atr{price} \approx \atr{budget\kern-1pt}}
    {\mathcal{D}_1}{\mathcal{D}_c}},
\end{align*}
where $\mathcal{D}_1$ stands for Table~\ref{tab:sales_a}
and $\mathcal{D}_c$ stands for the RDT of customers is in Table~\ref{tab:join}
(for brevity, some records are omitted). The RDT thus represents an answer to
query ``show deals for houses sold for (approximately) \val{\$200,000} with
(approximately) \val{1200} square feet and customers so that their budget is
similar to the house price''. Furthermore, we can obtain an RDT of best
agent-customer matching is we project the join onto \atr{AGENT} and \atr{NAME}:
\begin{align*}
  \project{\{\atr{AGENT},\atr{NAME}\}}
  {\join{\atr{price} \approx \atr{budget\kern-1pt}}
    {\mathcal{D}_1}{\mathcal{D}_c}}.
\end{align*}
The result of matching is in Table~\ref{tab:join_ab}\,(left). 
Due to our results, if we perform the same query with Table~\ref{tab:sales_b}
instead of Table~\ref{tab:sales_a}, the new result is guaranteed to be similar
with the obtained result at least to degree $0.98$. The result for
Table~\ref{tab:sales_b} is shown in Table~\ref{tab:join_ab}\,(right). 

\begin{table}[t]
  \centering
  \caption{Results of agent-customer matching for
    Table~\ref{tab:sales_a} and Table~\ref{tab:sales_b}}
  \label{tab:join_ab}
  \setlength{\tabcolsep}{4pt}
  \renewcommand{\arraystretch}{0.9}
  \begin{tabular}{|r||*{2}l|}
    \hline &
    \atr{AGENT} &
    \atr{NAME} \\
    \hline
    \rule{0pt}{8pt}%
    $0.91$ & \val{Brown} & \val{Grant} \\
    $0.89$ & \val{Brown} & \val{Evans} \\
    $0.89$ & \val{Brown} & \val{Finch} \\
    $0.88$ & \val{Clark} & \val{Grant} \\
    $0.86$ & \val{Clark} & \val{Evans} \\
    $0.84$ & \val{Clark} & \val{Finch} \\
    $0.74$ & \val{Davis} & \val{Evans} \\
    $0.72$ & \val{Davis} & \val{Grant} \\
    $0.66$ & \val{Davis} & \val{Finch} \\
    \hline
  \end{tabular}
  \qquad\qquad
  \begin{tabular}{|r||*{2}l|}
    \hline &
    \atr{AGENT} &
    \atr{NAME} \\
    \hline
    \rule{0pt}{8pt}%
    $0.91$ & \val{Brown} & \val{Grant} \\
    $0.90$ & \val{Clark} & \val{Grant} \\
    $0.89$ & \val{Brown} & \val{Evans} \\
    $0.89$ & \val{Brown} & \val{Finch} \\
    $0.88$ & \val{Clark} & \val{Evans} \\
    $0.86$ & \val{Clark} & \val{Finch} \\
    $0.75$ & \val{Davis} & \val{Evans} \\
    $0.73$ & \val{Davis} & \val{Grant} \\
    $0.67$ & \val{Davis} & \val{Finch} \\
    \hline
  \end{tabular}
\end{table}

\subsection{Tuple-Based Similarity and Further Topics}\label{sec:tbsim}
While the rank-based similarity from Section~\ref{sec:rnk} can be
sufficient in many cases, there are situations where one wants to consider
a similarity of RDTs based on ranks and (pairwise) similarity of tuples.
For instance, if we take
the RDT from Table~\ref{tab:sales_a} and make a new one by taking all tuples
(keeping their ranks) and increasing the prices by one dollar, we will come up
with an RDT which is, according to rank-based similarity, very different from
the original one. Intuitively, one would expect to have a high degree of
similarity of the RDTs because they differ only by a slight change in price.
This issue can be solved by considering the following tuple-based
degree of inclusion:
\begin{align}
  S^\approx(\mathcal{D}_1,\mathcal{D}_2) &=
  \textstyle\bigwedge_{t \in \mathrm{Tupl}(T)}
  \bigl(\mathcal{D}_1(t) \rightarrow
  \textstyle\bigvee_{\!t' \in \mathrm{Tupl}(T)}
  \bigl(\mathcal{D}_2(t') \otimes t \approx t'\bigr)\bigr),
  \label{def:SH}
\end{align}
where $t \approx t' = \textstyle\bigwedge_{y \in T}t(y) \approx_y t'(y)$
is a similarity of tuples $t$ and $t'$ over $T$, cf.~\cite{BeVy:Crmdfl}.
In a similar way as in~\eqref{eqn:ES}, we may define $E^\approx$
using $S^\approx$ instead of $S$.

\begin{remark}
  By an easy inspection, $S(\mathcal{D}_1,\mathcal{D}_2) \leq
  S^\approx(\mathcal{D}_1,\mathcal{D}_2)$, i.e. \eqref{def:SH} yields
  an estimate which is at least as high as \eqref{def:S} and analogously
  for $E$ and $E^\approx$. Note that \eqref{def:SH} has a natural meaning.
  Indeed, $S^\approx(\mathcal{D}_1,\mathcal{D}_2)$ can be understood as
  a degree to which the following statement is true: ``If $t$ belongs to
  $\mathcal{D}_1$, then there is $t'$ which is similar to $t$ and which
  belongs to $\mathcal{D}_2$''. Hence, $E^\approx(\mathcal{D}_1,\mathcal{D}_2)$
  is a degree to which for each tuple from $\mathcal{D}_1$ there is a similar
  tuple in $\mathcal{D}_2$ and \emph{vice versa.} If $\mathbf{L}$ is a two-element
  Boolean algebra and each $\approx_y$ is an identity, then
  $E^\approx(\mathcal{D}_1,\mathcal{D}_2) = 1$ if{}f
  $\mathcal{D}_1$ and $\mathcal{D}_2$ are identical (in the usual sense).
\end{remark}

For tuple-based inclusion (similarity) and for certain relational operations,
we can prove analogous preservation formulas as in Section~\ref{sec:rnk}.
For instance,
\begin{align}
  S^\approx(\mathcal{D}_1,\mathcal{D}'_1) \wedge S(\mathcal{D}_2,\mathcal{D}'_2)
  &\leq
  S^\approx(\mathcal{D}_1 \cup \mathcal{D}_2,\mathcal{D}'_1 \cup \mathcal{D}'_2),
  \label{eqn:t_inc_union}
  \\
  S^\approx(\mathcal{D}_1,\mathcal{D}'_1) \otimes S(\mathcal{D}_2,\mathcal{D}'_2)
  &\leq
  S^\approx(\mathcal{D}_1 \times \mathcal{D}_2,
  \mathcal{D}'_1 \times \mathcal{D}'_2),
  \label{eqn:t_inc_cart*}
  \\
  S^\approx(\mathcal{D},\mathcal{D}') &\leq
  S^\approx(\project{R}{\mathcal{D}}, \project{R}{\mathcal{D}'}).
\end{align}
On the other hand, similarity-based selection $\sigma_{y \approx d}$
(and, as a consequence, similarity-based join $\bowtie_{p \approx q}$)
does not preserve $S^\approx$ in general which can be seen as a technical
complication. This issue can be overcome by introducing a new type
of selection $\sigma^\approx_{y \approx d}$ which is
\emph{compatible} with $S^\approx$.
Namely, we can define
\begin{align}
  \bigl(\selectH{y \approx d}{\mathcal{D}}\bigr)(t) &=
  \textstyle\bigvee_{\!
    t' \in \mathrm{Tupl}(T)}
  \bigl(\mathcal{D}(t') \otimes t' \approx t \otimes t(y) \mathop{\approx_y}
  d\bigr).\label{eqn:selectH}
\end{align}
For this notion, we can prove that
$S^\approx(\mathcal{D},\mathcal{D}') \leq
S^\approx(\selectH{y \approx d}{\mathcal{D}},
\selectH{y \approx d}{\mathcal{D}'})$. Similar extension can be done
for any relational operation which does not preserve $S^\approx$ directly.
Detailed description of the extension is postponed to a full version
of the paper because of the limited scope.

\subsection{Unifying Approach to Similarity of RDTs}\label{sec:common}
In this section, we outline a general approach to similarity
of RDTs that includes both the approaches from the previous sections.
Interestingly, both \eqref{def:S} and \eqref{def:SH} have a common
generalization using truth-stressing hedges~\cite{Haj:MFL,Haj:Ovt}.
Truth-stressing hedges represent unary operations on complete residuated lattices
(denoted by $*$ ) that serve as interpretations of logical connectives like
``very true'', see~\cite{Haj:MFL}.
Two boundary cases of hedges are (i) identity, i.e.
$a^{\ast} = a$ ($a \in L$); (ii) globalization: $1^{\ast} = 1$,
and $a^{\ast} = 0$ if $a < 1$. The globalization~\cite{TaTi:Gist} is
a hedge which can be interpreted as ``fully true''.

Let ${}^*$ be truth-stressing hedge on $\mathbf{L}$. For RDTs
$\mathcal{D}_1,\mathcal{D}_2$ on $T$, we define the degree
$S^\approx_{\hh}(\mathcal{D}_1,\mathcal{D}_2)$ of inclusion of $\mathcal{D}_1$
in $\mathcal{D}_2$ (with respect to ${}^*$) by
\begin{align}
  S^\approx_\hh(\mathcal{D}_i,\mathcal{D}_j) &=
  \textstyle\bigwedge_{t \in \mathrm{Tupl}(T)}
  \bigl(\mathcal{D}_i(t) \rightarrow
  \textstyle\bigvee_{\!t' \in \mathrm{Tupl}(T)}
  \bigl(\mathcal{D}_j(t') \otimes (t \approx t')^*\bigr)\bigr).
  \label{eqn:S3}
\end{align}
Now, it is easily seen that for $\hh$ being the identity, \eqref{eqn:S3}
coincides with \eqref{def:SH}; if $\approx$ is separating
(i.e., $t_1 \approx t_2 = 1$ if{}f $t_1$ is identical to $t_2$)
and $\hh$ is the globalization,
\eqref{eqn:S3} coincides with \eqref{def:S}. Thus, both
\eqref{def:S} and \eqref{def:SH} are particular instances of \eqref{eqn:S3}
resulting by a choice of the hedge. Note that identity and globalization are
two borderline cases of hedges. In general, complete residuated lattices
admit other nontrivial hedges that can be used in \eqref{eqn:S3}.
Therefore, the hedge in~\eqref{eqn:S3} serves as a parameter that has an
influence on how much emphasis we put on the fact that two tuples are similar.
In case of globalization, we put full emphasis, i.e., the tuples are required
to be equal to degree $1$ (exactly the same if $\approx$ is separating). 

If we consider properties needed to prove analogous estimation formulas
for general $S^\approx_{\hh}$ as we did in case of $S$ and $S^\approx$, we
come up with the following important property:
\begin{align}
  (r\approx s)^\hh \otimes (s\approx t)^\hh &\leq (r\approx t)^\hh,
  \label{eqn:tr}
\end{align}
for every $r,s,t \in \mathrm{Tupl}(T)$ which can be seen as transitivity of
$\approx$ with respect to $\otimes$ and ${}^*$. Consider the following two
cases in which \eqref{eqn:tr} is satisfied:
\bgroup
\addtolength{\leftmargini}{2.1em}
\begin{itemize}\parskip=4pt
\item[Case 1:]
  $\hh$ is globalization and $\approx$ is separating.
  If the left hand side of \eqref{eqn:tr} is nonzero, then 
  $r \approx s = 1$ and $s \approx t = 1$. Separability implies $r=s=t$,
  i.e. $(r \approx t)^\hh = 1^\hh = 1$, verifying \eqref{eqn:tr}.
\item[Case 2:]
  $\approx$ is transitive. In this case, since
  $a^\hh\otimes b^\hh \leq (a\otimes b)^\hh$ (follows from properties of hedges
  by standard arguments), transitivity of $\approx$ and monotony of ${}^*$ yield
  $(r\approx s)^\hh \otimes (s\approx t)^\hh \leq
  ((r\approx s)\otimes (s\approx t))^\hh \leq (r\approx t)^\hh$.
\end{itemize}
\egroup
The following lemma shows that $S^\approx_\hh$ and consequently $E^\approx_\hh$
have properties that are considered natural for (degrees of) inclusion and 
similarity:

\begin{lemma}
  If $\approx$ satisfies \eqref{eqn:tr} with respect to ${}^*$ then
  \vspace*{-4pt}
  \bgroup
  \addtolength{\leftmargini}{4pt}
  \begin{itemize}\parskip=2pt
  \item[\rm (\emph{i})]
    $S^\approx_\hh$ is a reflexive and transitive $\mathbf{L}$-relation,
    i.e. an $\mathbf{L}$-quasiorder.
  \item[\rm (\emph{ii})]
    $E^\approx_\hh$ defined by $E^\approx_\hh(\mathcal{D}_1,\mathcal{D}_2) =
    S^\approx_\hh(\mathcal{D}_1,\mathcal{D}_2) \wedge
    S^\approx_\hh(\mathcal{D}_2,\mathcal{D}_1)$ is
    a reflexive, symmetric, and transitive $\mathbf{L}$-relation,
    i.e. an $\mathbf{L}$-equivalence.
  \end{itemize}
  \egroup
\end{lemma}
\begin{proof}
  The assertion follows from results in \cite[Section 4.2]{Bel:FRS} by taking
  into account that $\approx^\hh$ is reflexive, symmetric, and transitive
  with respect to $\otimes$.
  \qed
\end{proof}


\section{Conclusion and Future Research}
We have shown that an important fragment of relational operation in
similarity-based databases preserves various types of similarity.
As a result, similarity of query results based on these relational operations
can be estimated based on similarity of input data tables before the queries
are executed. Furthermore, the results of this paper have shown 
a desirable important property of the underlying similarity-based model of data:
slight changes in input data do not produce huge changes in query results.
Future research will focus on the role of particular relational operations
called similarity-based closures that play an important role in tuple-based
similarities of RDTs. An outline of results in this direction is presented
in \cite{BeUrVy:Scoqrisbd}.


\end{document}